\documentclass[twocolumn,showpacs,notitlepage,nofootinbib,floatfix,superscriptaddress]{revtex4-1} 

\usepackage{graphicx}
\usepackage{amssymb}
\usepackage{amsmath}
\usepackage{amsthm}
\usepackage{amsfonts}
\usepackage{mathrsfs}
\usepackage{color}
\usepackage{hyperref}
\usepackage{amscd}

\DeclareMathOperator{\im}{im}

\newcommand{\pr}[1]{\mathrm{pr}(#1)}
\newcommand{\suc}{\mathrm{succ}}
\newcommand{\face}[2]{{\Delta_{#1}(#2)}}
\DeclareMathOperator{\const}{\mathrm{const}}

\newcommand\pone{1.9\%}
\newcommand\ptwo{27.6\%}

\theoremstyle{lemma}
\newtheorem{lemma}{Lemma}

\begin{document}

\title{Three-dimensional color code thresholds via statistical-mechanical mapping}

\author{Aleksander Kubica}
\affiliation{Institute for Quantum Information \& Matter, California Institute of Technology,  Pasadena, CA 91125, USA}
\author{Michael E. Beverland}
\affiliation{Institute for Quantum Information \& Matter, California Institute of Technology,  Pasadena, CA 91125, USA}
\affiliation{Station Q Quantum Architectures and Computation Group, Microsoft Research, Redmond, WA 98052, USA}
\author{Fernando Brand\~ao}
\affiliation{Station Q Quantum Architectures and Computation Group, Microsoft Research, Redmond, WA 98052, USA}
\affiliation{Institute for Quantum Information \& Matter, California Institute of Technology,  Pasadena, CA 91125, USA}
\author{John Preskill}
\affiliation{Institute for Quantum Information \& Matter, California Institute of Technology,  Pasadena, CA 91125, USA}
\author{Krysta M. Svore}
\affiliation{Station Q Quantum Architectures and Computation Group, Microsoft Research, Redmond, WA 98052, USA}

\date{\today}

\begin{abstract}
Three-dimensional (3D) color codes have advantages for fault-tolerant quantum computing, such as protected quantum gates with relatively low overhead and robustness against imperfect measurement of error syndromes. 
Here we investigate the storage threshold error rates for bit-flip and phase-flip noise in the 3D color code on the body-centererd cubic lattice, assuming perfect syndrome measurements.
In particular, by exploiting a connection between error correction and statistical mechanics, we estimate the threshold for 1D string-like and 2D sheet-like logical operators to be $p^{(1)}_\mathrm{3DCC} \simeq \pone$ and $p^{(2)}_\mathrm{3DCC} \simeq \ptwo$.
We obtain these results by using parallel tempering Monte Carlo simulations to study the disorder-temperature phase diagrams of two new 3D statistical-mechanical models: the 4- and 6-body random coupling Ising models.
\end{abstract}

\maketitle

Some approaches to building scalable quantum computers are more practical than others due to their more favorable noise and resource requirements.   
The two-dimensional (2D) surface code approach \cite{Kitaev2003, Bravyi1998, Dennis2002} has very desirable features: (1) geometrically local syndrome measurements, (2) a high accuracy threshold and (3) fault-tolerant Clifford gates with low overhead.
Unfortunately, the surface code is {\it not} known to admit a (4) fault-tolerant non-Clifford gate with low overhead. 
The formidable qubit overhead cost of state distillation \cite{Bravyi2005, Fowler2012} for the necessary non-Clifford gate motivates the quest for alternatives to the surface code with all features (1)--(4).

Such alternatives may be sought in the general class of topological codes \cite{Kitaev2003, Bravyi1998, Levin2005, Bombin2006, Bombin2013book}, which includes the surface code as a special case. 
By definition, topological codes require only geometrically local syndrome measurements and tend to have high accuracy thresholds.
Topological codes often admit some fault-tolerant transversal gates (implemented by the tensor product of single-qubit unitaries), which have low overhead cost.
However, no quantum error-detecting code (whether topological or not), has a universal transversal encoded gate set \cite{Zeng2011, Eastin2009}.

Here we focus on the 3D topological color codes \cite{Bombin2007, Bombin2013} closely related to the 3D toric code \cite{Kubica2015}, which come in two types.
The stabilizer type has 1D string-like $Z$ and 2D sheet-like $X$ logical operators, and a logical non-Clifford gate $T = \text{diag}(1,e^{i \pi /4})$ is transversal.
In the subsystem type, there are 1D string-like $X$ and $Z$ dressed logical operators, and all logical Clifford gates are transversal.
Moreover, in the subsystem color code it is possible to reliably detect measurement errors in a single time step \cite{Bombin2015, Brown2015}.
By fault-tolerantly switching between the stabilizer and subsystem color codes \cite{Bombin2013,Kubica2015a}, one can combine the desirable features (1), (3) and (4).

In this work, we address feature (2) for the 3D color codes by finding thresholds $p^{(1)}_\mathrm{3DCC} \simeq \pone$ and $p^{(2)}_\mathrm{3DCC} \simeq \ptwo$ for phase-flip $Z$ and bit-flip $X$ noise, respectively.
Our results assume optimal decoders for independent $X$ and $Z$ noise with perfect measurements, and thereby give fundamental error-correction bounds against which efficient, but suboptimal decoders (such as that studied in \cite{Brown2015}) can be compared.
These thresholds are comparable to the analogous thresholds for the cubic lattice 3D toric code: $p^{(1)}_{\text{3DTC}} \simeq 3.3\%$ and $p^{(2)}_{\text{3DTC}} \simeq 23.5\%$ \cite{Ozeki1998,Hasenbusch2007,Ohno2004}, but 
compare unfavorably to $p_{\text{2DTC}} \simeq 10.9\%$ for the square lattice 2D toric code \cite{Honecker2001}.

\begin{figure}[b!]
\vspace*{-5pt}
\includegraphics[width=.95\columnwidth]{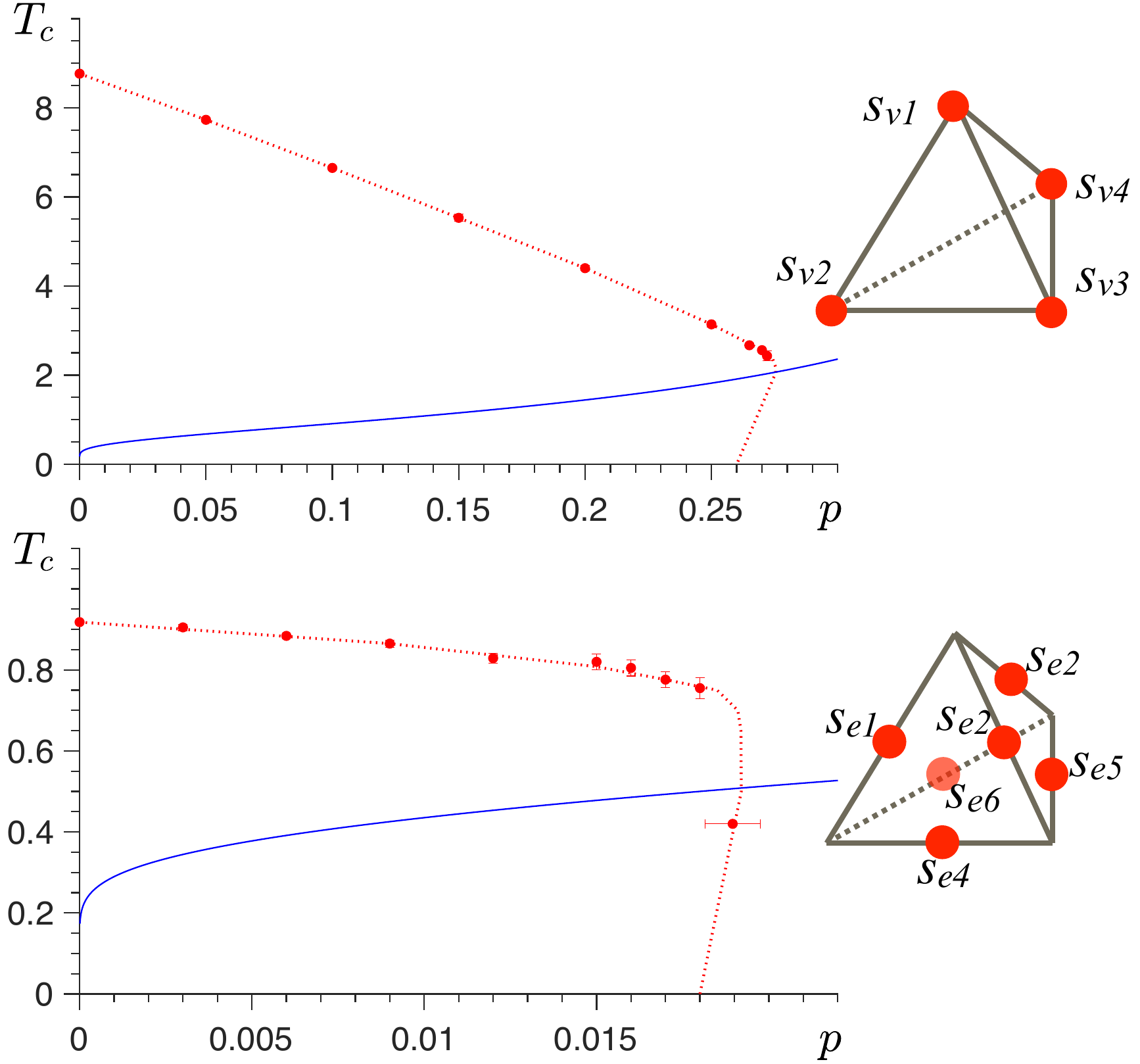}
\caption{The disorder-temperature $(p,T)$-phase diagrams of the $4$-body (top) and $6$-body (bottom) 3D random coupling Ising models.
Both models are defined on the 3D body-centered cubic lattice built of tetrahedra. The $4$- and $6$-body models have spins on vertices and edges, respectively. 
The error correction threshold $p_c$ can be found as the intersection of the Nishimori line (blue line) with the anticipated phase boundary (red dotted line).}
\label{fig_phasediag} 
\end{figure}

Our approach extends techniques known for other codes \cite{Dennis2002, Wang2002, Kovalev2013, Bombin2013book, Katzgraber2009, Bombin2012, Andrist2012} in order to relate the 3D color code thresholds to phase transitions in two new 3D statistical-mechanical models: the $4$- and $6$-body random coupling Ising models (RCIM).
We use large-scale parallel tempering Monte Carlo simulations \cite{Hukushima1996} and analyze specific heat, sublattice magnetization and Wilson loop operators to map the relevant parts of the disorder-temperature $(p,T)$-phase diagram; see Fig.~\ref{fig_phasediag}. 
The $6$-body RCIM is an example of a lattice gauge theory with a local (gauge) $\mathbb{Z}_2\times\mathbb{Z}_2$ symmetry, which makes this model both interesting and challenging to study.

\emph{3D stabilizer color code.--- }
Let $\mathcal{L}$ be a three-dimensional lattice built of tetrahedra such that its vertices are $4$-colorable, i.e., vertices connected by an edge are of different colors. 
An example of such a lattice is the body-centered cubic (bcc) lattice obtained from two interleaved cubic lattices; see Fig.~\ref{fig_bcc_lattices}(b). 
We denote by $\face{i}{\mathcal{L}}$ the set of all $i$-simplices of $\mathcal{L}$. 
Then, $0$-simplices of $\mathcal{L}$ are vertices, $1$-simplices are edges, etc. We place one qubit at every tetrahedron $t\in\face{3}{\mathcal{L}}$. For every vertex $v\in\face{0}{\mathcal{L}}$ and edge $e\in\face{1}{\mathcal{L}}$ we define operators $S_X(v)$ and $S_Z(e)$ to be the product of either Pauli $X$ or $Z$ operators on qubits identified with tetrahedra in the neighborhood of the vertex $v$ or edge $e$, namely
\begin{equation}
S_X(v) = \prod_{\substack{t\in\face{3}{\mathcal{L}}\\ t\supset v}} X(t),\quad
S_Z(e) = \prod_{\substack{t\in\face{3}{\mathcal{L}}\\ t\supset e}} Z(t).
\label{eq_generators}
\end{equation}
The 3D stabilizer \cite{subsystem}
color code is defined by specifying its stabilizer group \cite{Gottesman1996}
\begin{equation}
\mathcal{S} = \langle S_X(v), S_Z(e) | v\in\face{0}{\mathcal{L}}, e\in\face{1}{\mathcal{L}} \rangle. 
\end{equation}
Using the colorability condition one can show that $\mathcal{S}$ is an Abelian subgroup of the Pauli group not containing $-I$. The code space is the $+1$ eigenspace of all elements of $\mathcal{S}$ and the lowest-weight logical $X$ and $Z$ operators of the 3D color code are 2D sheet-like and 1D string-like objects; see Fig.~\ref{fig_bcc_lattices}(a).
In general, the color code can be defined in $d\geq 2$ dimensions on a lattice, provided it is a $(d+1)$-colorable simplical $d$-complex \cite{Kubica2015a}.

\begin{figure}[ht!]
\includegraphics[width=.94\columnwidth]{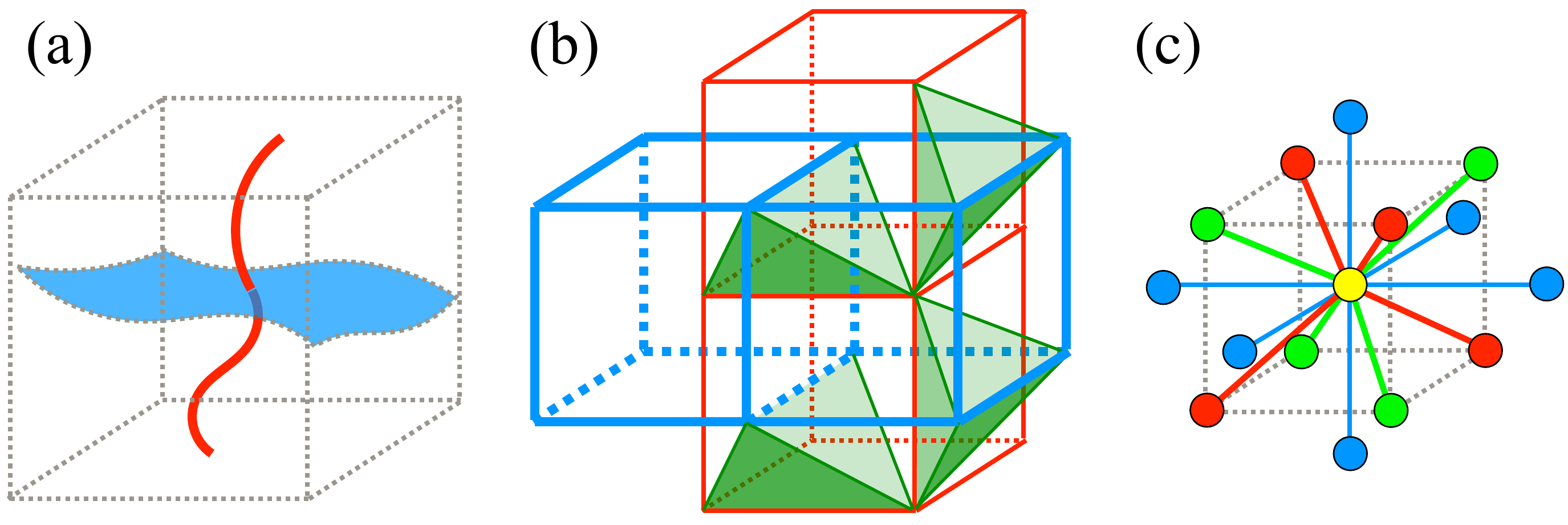}
\caption{(a) The 3D stabilizer color code has both 1D string-like (red) and 2D sheet-like (blue) logical operators. (b) The bcc lattice can be constructed starting from two interleaved cubic lattices (red and blue) and filling in tetrahedra (green). Not all tetrahedra are depicted. (c) The neighborhood of any vertex in the bcc lattice looks the same --- every vertex belongs to 24 edges, 36 triangular faces and 24 tetrahedra. The bcc lattice is $4$-colorable, i.e., every vertex is colored in red, green, blue or yellow, and no two neighboring vertices are of the same color.}
\label{fig_bcc_lattices} 
\end{figure}

\emph{Error correction in CSS codes.--- }
Since the color code is a CSS code \cite{Calderbank1997}, we choose to separately correct $X$- and $Z$-type errors, which simplifies the discussion.
We also assume perfect measurements.
For concreteness, we focus on $X$-error correction; $Z$-errors can be analyzed analogously \cite{xzreplacement}.

The set of all $Z$-type stabilizers which return $-1$ measurement outcomes is called a $Z$-type syndrome. Note that any nontrivial $Z$-syndrome signals the presence of some $X$-errors in the system. Correction of $X$-errors in a CSS code can be succinctly described by introducing a chain complex \cite{Freedman2001, Delfosse2014}
\begin{equation} 
\begin{array}{ccccc}
C_2 & \xrightarrow{\partial_2} & C_1 & \xrightarrow{\partial_1} & C_0\\
X\textrm{-stabilizers} & & \mathrm{qubits} & & Z\textrm{-stabilizers}
\end{array}
\label{eq_chain}
\end{equation}
where $C_2$, $C_1$ and $C_0$ are vector spaces over $\mathbb{Z}_2$ with bases $\mathcal{B}_2$ = $X$-stabilizer generators, $\mathcal{B}_1$ = physical qubits and $\mathcal{B}_0$ = $Z$-stabilizer generators, respectively.
The linear maps $\partial_2$ and $\partial_1$, called boundary operators, are chosen in such a way that the support of any $X$-stabilizer $\omega\in C_2$ is given by $\partial_2 \omega$, and the $Z$-syndrome corresponding to any $X$-error $\epsilon\in C_1$ can be found as $\partial_1 \epsilon$.
Note that $\partial_1\circ\partial_2 = 0$, since any $X$-stabilizer has trivial $Z$-syndrome.
One can think of the boundary operators as parity-check matrices $H^T_X$ and $H_Z$ of the CSS code.
In the case of the 3D color code, $C_2$, $C_1$, $C_0$ are generated by vertices, tetrahedra, and edges respectively, i.e., $\mathcal{B}_2 =\face{0}{\mathcal{L}}$, $\mathcal{B}_1 =\face{3}{\mathcal{L}}$ and $\mathcal{B}_0 =\face{1}{\mathcal{L}}$. 
The boundary operators are defined to be
$\partial_2 v = \sum_{\face{3}{\mathcal{L}}\ni t\supset v} t$ and $\partial_1 t = \sum_{\face{1}{\mathcal{L}}\ni e\subset t} e$ for any $v\in\face{0}{\mathcal{L}}$ and $t\in\face{3}{\mathcal{L}}$.

Let $\epsilon, \varphi\in C_1$ be two $X$-errors with the same $Z$-syndrome, $\partial_1 \epsilon = \partial_1 \varphi$. We say that $\epsilon$ and $\varphi$ are equivalent iff they differ by some $X$-stabilizer $\omega\in C_2$, namely $\epsilon+\varphi = \partial_2 \omega$.
To correct errors, we need a decoder --- an algorithm which takes the $Z$-syndrome $\sigma\in C_0$ as an input and returns a $Z$-correction $\varphi$ which will restore all $X$-stabilizers to have $+1$ outcomes, i.e., $\partial_1 \varphi = \sigma$. The decoder succeeds iff the actual error $\epsilon$ and the correction $\varphi$ are equivalent. An optimal decoder finds a representative $\varphi$ of the most probable equivalence class of errors $\overline\varphi = \{ \varphi + \partial_2 \omega | \forall\omega\in C_2\}$.

\emph{Statistical-mechanical models.--- }
In this section, we provide a brief derivation of the connection between optimal error-correction thresholds and phase transitions \cite{Dennis2002, Wang2002, Kovalev2013, Bombin2013book,Katzgraber2009, Bombin2012, Andrist2012}. 
In particular, we derive two new statistical-mechanical models relevant for the 3D color code.

We assume bit-flip noise, i.e., every qubit is independently affected by Pauli $X$ error with probability $p$. The probability of an $X$-error $\epsilon\in C_1$ affecting the system is
\begin{equation}
\pr{\epsilon}= \prod_{j \in \mathcal{B}_1} p^{[\epsilon]_j} (1-p)^{1-[\epsilon]_j} \propto \left(\frac{p}{1-p}\right )^{ \sum_{j \in \mathcal{B}_1} [\epsilon]_j},
\label{eq_prerror}
\end{equation}
where $[\epsilon]_j \in \mathbb{Z}_2$ denotes the $j$ coefficient of $\epsilon$ in the $\mathcal{B}_1$ basis, $\epsilon = \sum_{j\in\mathcal{B}_1} [\epsilon]_j j$.

For a general CSS code family with the chain complex in Eq.~(\ref{eq_chain}), the $X$-error correction threshold is the largest $p_c$ such that for all $p<p_c$ the probability of successful decoding goes to 1 in the limit of infinite system size
\begin{equation}
\pr\suc = \sum_{\epsilon\in C_1} \pr{\epsilon} \pr{\suc | \epsilon} \rightarrow 1.
\end{equation}
With the optimal decoder, the conditional probability $\pr{\suc | \epsilon}$ equals 1 if $\epsilon$ belongs to the most probable error equivalence class consistent with the syndrome $\partial_1 \epsilon$, and 0 otherwise. The probability of equivalence class $\overline{\epsilon}$ is
\begin{equation}
\pr{\overline\epsilon} =  \sum_{\omega\in C_2}\pr{\epsilon + \partial_2\omega} \propto  \sum_{\omega\in C_2} e^{-2\beta(p) \sum_{j \in \mathcal{B}_1} [\epsilon + \partial_2\omega]_j},
\label{eq_partition_deriv}
\end{equation}
where we use Eq.~(\ref{eq_prerror}) and introduce
\begin{equation}
\beta(p) = -\frac{1}{2}\log \frac{p}{1-p}.
\label{eq_nishimori}
\end{equation}
To rewrite Eq.~(\ref{eq_partition_deriv}), we use 
$[\partial_2 \omega ]_j \equiv \sum_{i\in\mathcal{B}_2 \wedge \partial_2 i \ni j} [\omega]_i \mod 2$
and $1-2[\epsilon + \partial_2\omega]_j= (-1)^{[\epsilon]_j}(-1)^{[\partial_2\omega]_j} = (-1)^{[\epsilon]_j}\prod_{i\in\mathcal{B}_2 \wedge \partial_2 i \ni j}(-1)^{[\omega]_i}$.
By introducing new (classical spin) variables $s_i = (-1)^{[\omega]_i}$ for all $i\in\mathcal{B}_2$, we can replace the sum over $\omega \in C_2$ in Eq.~(\ref{eq_partition_deriv}) by a sum over different configurations $\{s_i = \pm 1 \}$, yielding
\begin{equation}
\pr{\overline\epsilon} \propto \sum_{\{s_i = \pm 1\}} e^{-\beta(p) H_\epsilon (\{ s_i\})}
\label{eq_partition},
\end{equation}
where we introduce the Hamiltonian
\begin{equation}
H_\epsilon (\{ s_i\}) = -\sum_{j \in \mathcal{B}_1}  (-1)^{[\epsilon]_j} \prod_{\substack{i\in\mathcal{B}_2\\ [\partial_2 i]_j = 1}} s_i.
\label{eq_rcim}
\end{equation}

We define the random coupling Ising model (RCIM) to be a classical spin $s_i = \pm 1$ random model with quenched couplings $(-1)^{[\epsilon]_j}$ described by $H_\epsilon (\{ s_i\})$ in Eq.~(\ref{eq_rcim}).
The RCIM has two independent parameters: disorder strength $p$ (i.e., the probability of negative couplings) and inverse temperature $\beta$. 
The partition function of the RCIM with disorder $\epsilon$ at temperature $\beta^{-1}$ is given by
\begin{equation}
Z_\epsilon (\beta) = \sum_{\{s_i = \pm 1\}} e^{-\beta H_\epsilon (\{ s_i\})}.
\label{eq_partition}
\end{equation}
Note that for the proportionality $\pr{\overline\epsilon} \propto Z_\epsilon (\beta)$ in Eq.~(\ref{eq_partition}) to hold one requires $\beta = \beta(p)$.

For the 3D color code, Eq.~(\ref{eq_rcim}) leads to the following two new statistical-mechanical models
\begin{eqnarray}
H^{X}_{\epsilon} (\{s_v\}) &=& - \sum_{t\in\face{3}{\mathcal{L}}} (-1)^{[\epsilon]_t} \raisebox{-10pt}{\includegraphics[height = 25pt]{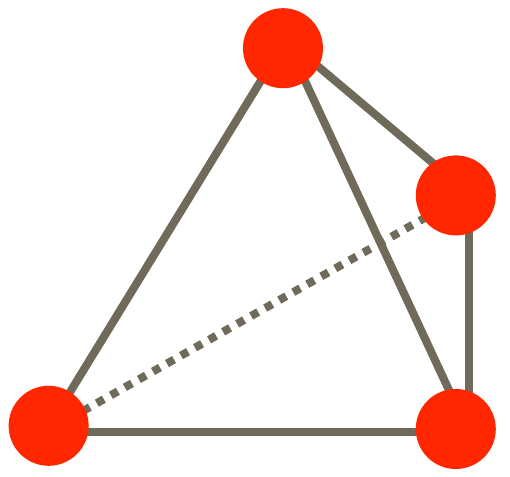}},
\label{eq_ham_4body}\\
H^Z_{\epsilon} (\{s_e\}) &=& - \sum_{t\in\face{3}{\mathcal{L}}} (-1)^{[\epsilon]_t} \raisebox{-10pt}{\includegraphics[height = 25pt]{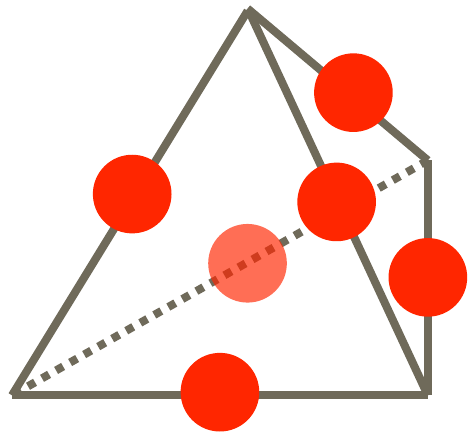}},
\label{eq_ham_6body}
\end{eqnarray}
relevant to correction of $X$- and $Z$-errors, respectively.
Note that $H^{X}_{\epsilon} (\{s_v\})$ (respectively $H^{Z}_{\epsilon} (\{s_e\})$) contains 4-body (6-body) terms, which are products of vertex (edge) spins of every tetrahedron.
We observe that for $p=0$, i.e., the case with no disorder, these two models are self-dual in the sense that the low-temperature expansion of each model matches the high-temperature expansion of the other \cite{Wegner1971}; see the Supplemental Material.

The Hamiltonian in Eq.~(\ref{eq_rcim}) determines a thermal ensemble of excitations in the statistical mechanical model. For $H^{X}_{\epsilon} (\{s_v\})$ the excitations are 2D domain walls residing on a set of tetrahedra $\varphi=\epsilon + \partial_2 \omega\in C_1$, where these walls terminate at the edges contained in $\partial_1\varphi=\partial_1 \epsilon \in C_0$. In the color code, this ensemble of domain walls corresponds to the ensemble of possible $X$-errors which generate the same error syndrome as  $\epsilon$, and the Boltzmann weight of a wall configuration coincides with the probability of the corresponding $X$-error configuration $\varphi$. Likewise, for  $H^{Z}_{\epsilon} (\{s_e\})$ the excitations are 1D strings terminating at vertices in $\partial_1\epsilon$, corresponding to $Z$-errors which generate the same error syndrome as $\epsilon$.

To determine the storage threshold for the 3D color code, we investigate the disorder-temperature $(p,T)$-phase diagram of the RCIM in Eq.~(\ref{eq_rcim}). In the ordered phase, large fluctuations of domain walls (or strings) are suppressed \cite{Dennis2002}, and the free energy cost
\begin{equation}
\Delta_\lambda(\epsilon) = -\log Z_{\epsilon+\lambda}(\beta) +\log Z_{\epsilon}(\beta)
\end{equation}
of introducing any non-trivial domain wall $\lambda\in\ker\partial_1 \setminus \im\partial_2$ to the system at inverse temperature $\beta$ with disorder $\epsilon$ should diverge in the limit of infinite system size when averaged over all disorder configurations
\begin{equation}
\langle \Delta_\lambda \rangle = \sum_{\epsilon\in C_1} \pr\epsilon \Delta_\lambda(\epsilon) \rightarrow \infty.
\end{equation}
Correspondingly, in the color code, the error $\varphi$ produces a syndrome $\partial_1\varphi$ which points to a unique equivalence class $\overline\varphi$, so that the syndrome can be decoded successfully with high probability.
Indeed, we show in the Supplemental Material, $\pr\suc\rightarrow 1$ for the error rate $p$ implies $\langle \Delta_\lambda \rangle \rightarrow \infty$ for the RCIM at inverse temperature $\beta(p)$ and disorder strength $p$. Thus, by finding the critical point along the line defined by Eq.~(\ref{eq_nishimori}) (the Nishimori line \cite{Nishimori1981}) we obtain the threshold value $p_c$.

\emph{Phase diagram.---}
We describe how to map out the $(p,T)$-phase diagrams of the two RCIMs, $H^{X}_{\epsilon} (\{s_v\})$ and $H^{Z}_{\epsilon} (\{s_e\})$.
The discontinuity in energy density across a first order phase transition allows for straightforward identification of the phase boundary in the regime of low disorder.
However, more reliable order parameters are required to probe a (higher-order) phase transition close to the critical point on the Nishimori line.
Moreover, an appropriate order parameter takes symmetries of the model into account. Note that flipping a subset of spins $\{s_i\}_{i\in I}$, i.e., $s_i \mapsto -s_i$ for $i\in I$, is a symmetry if it leaves the Hamiltonian describing the model invariant.

\begin{figure*}
\includegraphics[width=\textwidth]{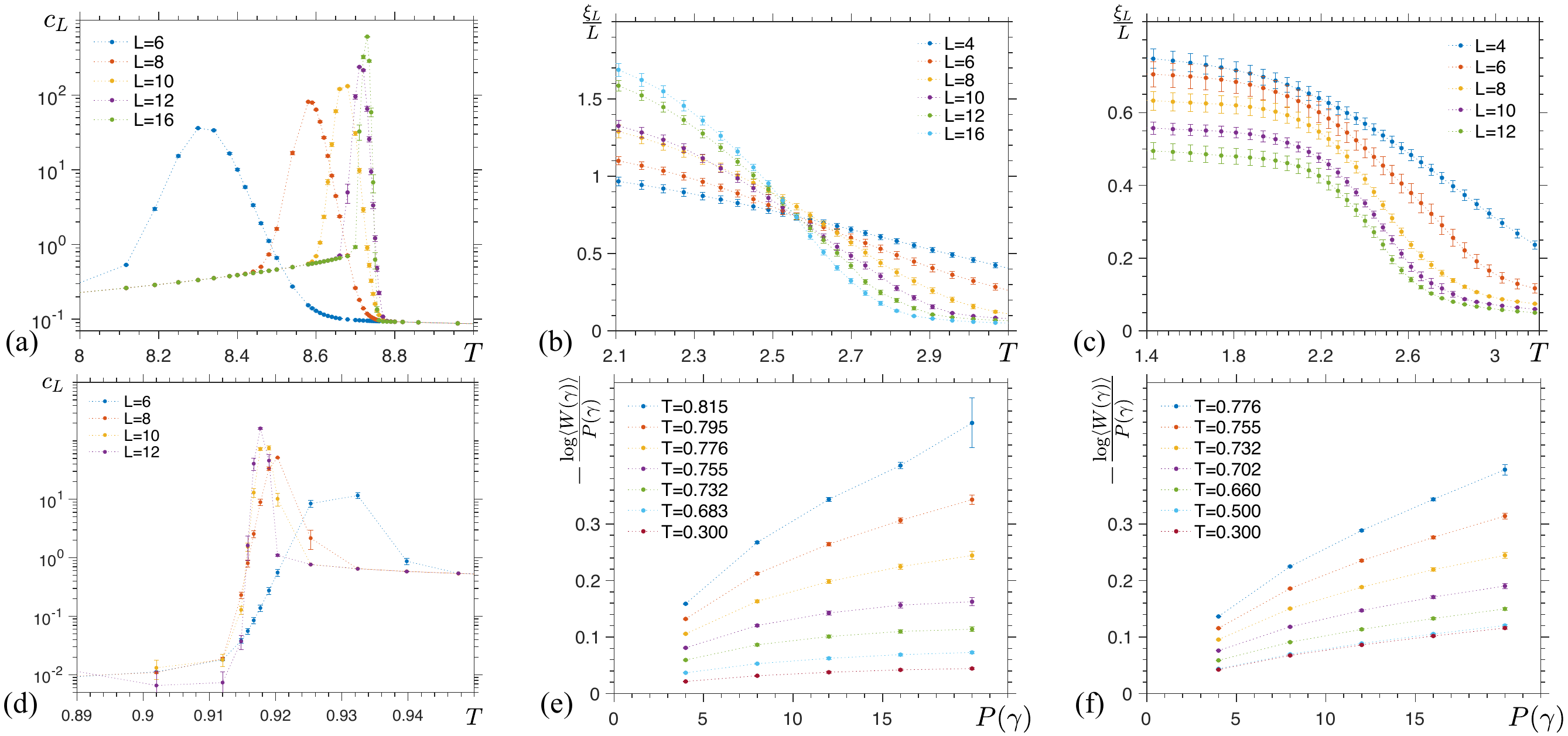}
\caption{
Results for the 3D $4$-body (a)-(c) and $6$-body (d)-(f) RCIM.
By finding the peak positions of specific heat $c_L$ for different system sizes $L$ and exploiting finite-size scaling we estimate for $p=0$ the critical temperature of a phase transition in (a) and (d) to be $T_c = 8.77(2)$ and $T_c=0.918(3)$.
(b) For $p=0.27$ we identify $T_c=2.56(4)$ as the intersection of normalized spin-spin correlation functions $\xi_L /L$ for different system sizes $L$. (c) For $p=0.28$ there is no indication of a phase transition.
In (e) and (f) we check if the Wilson loop operator $W(\gamma)$ satisfies the perimeter law by plotting $-\log \langle W(\gamma)\rangle /P(\gamma)$ as a function of perimeter $P(\gamma)$ of the square loop $\gamma$ for different temperatures $T$.
(e) For $p=0.018$ we see a change of scaling as the system undergoes a phase transition at $T = 0.75(3)$.
(f) For $p=0.021$ there is no indication of a phase transition.}
\label{fig_numerics} 
\end{figure*}

The $4$-body RCIM in Eq.~(\ref{eq_ham_4body}) has a global $\mathbb{Z}_2\times\mathbb{Z}_2\times\mathbb{Z}_2$ symmetry.
An example of a symmetry operation is a simultaneous flip of vertex spins on all red and blue vertices, since it leaves every term of $H^{X}_{\epsilon} (\{s_v\})$ unchanged.
Due to this symmetry, the total magnetization is not a good order parameter; however the sublattice magnetization of spins of a single color is. 
Instead of using the sublattice magnetization directly, more precise estimations are obtained by considering the finite-size scaling of the spin-spin correlation function \cite{Palassini1999}. 
Near the phase transition, for fixed disorder strength $p$ and temperatures $T$ close to the critical temperature $T_c(p)$, the correlation length $\xi_L$ is expected to scale as
\begin{equation}
\xi_L(p,T) /L \sim f(L^{1/\nu} (T-T_c(p))),
\label{eq_correlation}
\end{equation}
where $L$ is the linear system size, $f$ is a scaling function and $\nu$ is the correlation length critical exponent \cite{Goldenfeld1992}.
We can estimate $T_c(p)$ by plotting $\xi_L(p,T)/L$ as a function of temperature $T$ for different system sizes $L$ and finding their crossing point; see Fig.~\ref{fig_numerics}(a)(b).
If no crossing is observed, then we conclude that there is no phase transition.

The 6-body RCIM in Eq.~(\ref{eq_ham_6body}) describes a lattice gauge theory with a local $\mathbb{Z}_2\times\mathbb{Z}_2$ symmetry.
An example of a symmetry operation is a flip of edge spins on edges from a single yellow vertex to all neighboring red and blue vertices; see Fig.~\ref{fig_bcc_lattices}(c).
Due to Elitzur's theorem \cite{Elitzur1975}, the gauge symmetry rules out existence of any local order parameter.
We define a Wilson loop operator \cite{Wilson1974, Kogut1979}
\begin{equation}
W(\gamma) = \prod_{e\in\gamma} s_e,
\label{eq_wilson}
\end{equation}
to be a product of edge spins along a loop $\gamma\subset\face{1}{\mathcal{L}}$.
For $W(\gamma)$ to be gauge-invariant the loop $\gamma$ can only be composed of edges connecting vertices of two (out of four possible) colors. 
The phase transition can be identified by analyzing scaling of the thermal expectation value of $W(\gamma)$ averaged over different disorder configurations
\begin{equation}
\langle W(\gamma) \rangle = \sum_{\epsilon\subset\face{3}{\mathcal{L}}} \pr\epsilon \sum_{\{s_e\}} W(\gamma)\frac{e^{-\beta H^Z_{\epsilon}(\{s_e\})}}{Z_\epsilon(\beta)}.
\end{equation}
Namely, in the limit of large square loops \cite{Creutz1979, Wang2002, Ohno2004}, $-\log \langle W(\gamma) \rangle$ scales linearly with the loop's perimeter $P(\gamma)$ in the ordered (Higgs) phase, whereas in the disordered (confinement) phase it scales linearly with the minimum area $A(\gamma)$ enclosed by $\gamma$; see Fig.~\ref{fig_numerics}(d)-(f).

We find the $(p,T)$-phase diagrams of the $4$- and $6$-body RCIMs by performing Monte Carlo simulations with parallel tempering technique \cite{Hukushima1996}; see Fig.~\ref{fig_phasediag}.
We test equilibration of the system by a logarithmic binning of the data.
Since we can simulate only finite-size systems, a careful analysis of finite-size effects is necessary.
Parameters of numerical simulations and additional details are provided in the Supplemental Material.

\emph{Discussion.--- }
We have found 3D stabilizer color code thresholds for phase-flip $Z$ and bit-flip $X$ noise models with optimal decoding and perfect measurements: $p^{(1)}_\mathrm{3DCC} \simeq \pone$ and $p^{(2)}_\mathrm{3DCC}  \simeq \ptwo$.
The $X$-stabilizers detecting $Z$-errors are the same for the 3D stabilizer and subsystem color codes. 
Since the subsystem code is symmetric under the exchange of $X$- and $Z$-generators, its phase- and bit-flip thresholds are the same and equal to $p^{(1)}_\mathrm{3DCC}$ of the stabilizer color code on the same lattice family.
The 3D color code threshold \cite{comparison} with the (efficient) clustering decoder $p^{(1)}_{\text{clust}} \simeq 0.46 \%$ \cite{Brown2015} is about a fourth of $p^{(1)}_\mathrm{3DCC}$, justifying a search for efficient color-code decoders with performance as close to optimal as for 2D surface and color codes.

We hope that our work initiates and motivates a careful study of the 3D random coupling Ising models and their phase diagrams.
We conjecture the existence of a spin-glass phase \cite{Binder1986} in the $6$-body RCIM, corresponding to a regime with intermediate noise strength in which memory lifetime with non-optimal decoders is polynomial rather than exponential in the system size.

A future extension of this work might incorporate measurement errors in the analysis. 
This would require the study of 4D random coupling models and thus use more computational resources.
If successful, this research program could provide a deeper understanding of single-shot error correction \cite{Bombin2015, Brown2015} from the standpoint of statistical mechanics.

We thank R. Andrist, H. Bomb\'in, N. Delfosse, L. Pryadko, B. Yoshida and I. Zintchenko for helpful discussions. AK would like to thank the QuArC group for their hospitality during a summer internship. We acknowledge funding provided by the Institute for Quantum Information and Matter, an NSF Physics Frontiers Center (NFS Grant PHY-1125565) with support of the Gordon and Betty Moore Foundation (GBMF-12500028).

\appendix
\onecolumngrid

\section{Duality of models for zero disorder}

We already mentioned that the $4$- and $6$-body RCIM described by Eqs.~(\ref{eq_ham_4body})~and~(\ref{eq_ham_6body}) are dual for $p=0$, i.e., the case with no disorder . Here we say that two models are dual if the low-temperature expansion of the partition function of one model matches the high-temperature expansion of the partition function of the other and vice versa \cite{Wegner1971}. We observe that for any CSS code, the two statistical-mechanical models relevant for correction of $X$- and $Z$-errors are always dual for $p=0$. In particular, if there is only one phase transition in the first model at temperature $T_c^X$, then there is a unique phase transition in the dual model at temperature
\begin{equation}
T_c^Z = -\frac{2}{T_c^X} \log \tanh \frac{1}{T_c^X}.
\label{eq_dualtemp}
\end{equation}
This serves as a consistency check for our results. Indeed, for zero disorder $p=0$ the critical temperatures $T_c^X = 8.77(2)$ and $T_c^Z = 0.918(3)$ for the $4$- and $6$-body RCIM are related according to Eq.~(\ref{eq_dualtemp}) within the statistical uncertainty.

\section{Proof of implication}

Here we show that successful decoding implies diverging average energy cost of introducing any non-trivial domain wall.
We used this fact in the derivation of statistical-mechanical models to relate the threshold $p_c$ of optimal error correction to the critical point $p_N$ on the Nishimori line.
Note that this implication allows us to only infer that $p_c \leq p_N$.
However, we expect that successful decoding be possible throughout the ordered phase and thus these two values should coincide.

\begin{lemma}
Consider a CSS code described by the chain complex in Eq.~(\ref{eq_chain}). Let $H_1 = \ker\partial_1 / \im\partial_2$ be the first homology group of finite cardinality, $|H_1| < \infty$.
If the probability of successful optimal $X$-error correction goes to 1 in the limit of infinite system size
\begin{equation}
\pr\suc = \sum_{\epsilon\in C_1} \pr{\epsilon} \pr{\suc | \epsilon} \rightarrow 1,
\end{equation}
then the average free energy cost of introducing any non-trivial domain wall $\lambda\in\ker \partial_1 \setminus \im \partial_2$ diverges
\begin{equation}
\langle \Delta_{\lambda} \rangle = \sum_{\epsilon\in C_1} \pr{\epsilon} \Delta_\lambda(\epsilon) \rightarrow \infty.
\end{equation}
\end{lemma}

\begin{proof}
Let $\overline\epsilon = \{ \epsilon + \partial_2\omega | \omega \in C_2\}$ denote the equivalence class of errors for $\epsilon\in C_1$ and $\mathcal{E} = \{ \overline\epsilon,\ldots \}$ be the set of all equivalence classes.
We define a representative of the most probable equivalence class of errors consistent with the syndrome $\sigma \in C_0$ to be
\begin{equation}
\rho(\sigma) = \arg\max_{\substack{\rho \in C_1\\ \partial_1 \rho = \sigma}} \pr{\overline\rho} .
\end{equation}
The conditional probability of successful decoding using the optimal (maximum likelihood) decoder is given by
\begin{equation}
\pr{\suc |\epsilon} = 
\begin{cases}
1 \mathrm{\ if \ } \epsilon\in\overline{\rho(\partial_1 \epsilon)},\\
0 \mathrm{\ otherwise.}
\end{cases}
\end{equation}
Thus, we have
\begin{equation}
\pr{\suc} = \sum_{\epsilon\in C_1} \pr{\epsilon} \pr{\suc | \epsilon} = \sum_{\sigma\in \im\partial_1} \pr{\overline{\rho(\sigma)}}.
\end{equation}
By rewriting the sum over all equivalence classes of errors $\overline\epsilon\in\mathcal{E}$ as the sum over all possible syndromes $\sigma\in\im \partial_1$ and different representatives $\lambda'\in H_1$ of the homology group we arrive at
\begin{equation}
1 = \sum_{\overline\epsilon\in\mathcal{E}} \pr{\overline \epsilon} = \sum_{\sigma\in\im\partial_1}\sum_{\lambda' \in H_1} \pr{\overline{\rho(\sigma) + \lambda'}}
= \pr{\suc} + \sum_{\sigma\in\im\partial_1}\sum_{0\neq\lambda' \in H_1} \pr{\overline{\rho(\sigma) + \lambda'}}.
\end{equation}

We want to show two inequalities
\begin{equation}
\pr{\suc} \geq \sum_{\overline\epsilon\in\mathcal{E}} \pr{\overline\epsilon} \frac{\pr{\overline{\epsilon}}}{\sum_{\lambda'\in H_1} \pr{\overline{\epsilon+\lambda'}}} \geq 2\pr{\suc} - 1.
\label{eq_prsuc}
\end{equation}
In order to show the first inequality~(\ref{eq_prsuc}) note that
\begin{eqnarray}
\pr{\suc} &=& \sum_{\sigma\in\im\partial_1} \sum_{\lambda''\in H_1}\pr{\overline{\rho(\sigma)+\lambda''}}
\frac{\pr{\overline{\rho(\sigma)}}}{\sum_{\lambda'\in H_1}\pr{\overline{\rho(\sigma)+\lambda'}}}\\
&\geq& \sum_{\sigma\in\im\partial_1} \sum_{\lambda''\in H_1}\pr{\overline{\rho(\sigma)+\lambda''}}
\frac{\pr{\overline{\rho(\sigma)+\lambda''}}}{\sum_{\lambda'\in H_1}\pr{\overline{\rho(\sigma)+\lambda'}}}
= \sum_{\overline\epsilon\in\mathcal{E}} \pr{\overline\epsilon} \frac{\pr{\overline\epsilon}}{\sum_{\lambda'\in H_1}\pr{\overline{\epsilon+\lambda'}}},
\end{eqnarray}
where we use $\pr{\overline{\rho(\sigma)}} \geq \pr{\overline{\rho(\sigma)+\lambda''}}$ for all $\sigma\in\im\partial_1$ and $\lambda''\in H_1$.
The second inequality~(\ref{eq_prsuc}) follows from 
\begin{eqnarray}
\pr{\suc} &=& \sum_{\sigma\in\im\partial_1} \pr{\overline{\rho(\sigma}})
= \sum_{\sigma\in\im\partial_1} \pr{\overline{\rho(\sigma)}} \frac{\pr{\overline{\rho(\sigma)}}}{\sum_{\lambda'\in H_1} \pr{\overline{\rho(\sigma)+\lambda'}}}
+ \sum_{\sigma\in\im\partial_1} \pr{\overline{\rho(\sigma)}} \frac{\sum_{0\neq\lambda'\in H_1} \pr{\overline{\rho(\sigma)+\lambda'}}}{\sum_{\lambda'\in H_1} \pr{\overline{\rho(\sigma)+\lambda'}}}\quad\quad\\
&\leq& \sum_{\sigma\in\im\partial_1}\sum_{\lambda''\in H_1} \pr{\overline{\rho(\sigma)+\lambda''}}
\frac{\pr{\overline{\rho(\sigma)+\lambda''}}}{\sum_{\lambda'\in H_1} \pr{\overline{\rho(\sigma)+\lambda'}}}
+ \sum_{\sigma\in\im\partial_1} \sum_{0\neq\lambda'\in H_1} \pr{\overline{\rho(\sigma)+\lambda'}}\\
&=& \sum_{\overline\epsilon\in\mathcal{E}} \pr{\overline\epsilon} \frac{\pr{\overline{\epsilon}}}{\sum_{\lambda'\in H_1} \pr{\overline{\epsilon+\lambda'}}} + (1 - \pr{\suc}).
\end{eqnarray}

If $\pr{\suc} \rightarrow 1$, then from inequalities~(\ref{eq_prsuc}) we infer that
\begin{equation}
\sum_{\overline\epsilon\in\mathcal{E}} \pr{\overline\epsilon} \frac{\pr{\overline{\epsilon}}}{\sum_{\lambda'\in H_1} \pr{\overline{\epsilon+\lambda'}}} \rightarrow 1,
\end{equation}
and thus for $\lambda \in\ker\partial_1\setminus\im\partial_2$ we have
\begin{equation}
\sum_{\overline\epsilon\in\mathcal{E}} \pr{\overline\epsilon} \frac{\pr{\overline{\epsilon+\lambda}}}{\sum_{\lambda'\in H_1} \pr{\overline{\epsilon+\lambda'}}} \rightarrow 0.
\end{equation}
In the last step we used the following inequalities
\begin{equation}
0\leq \sum_{\overline\epsilon\in\mathcal{E}} \pr{\overline\epsilon} \frac{\pr{\overline{\epsilon+\lambda}}}{\sum_{\lambda'\in H_1} \pr{\overline{\epsilon+\lambda'}}} 
\leq \sum_{\overline\epsilon\in\mathcal{E}} \pr{\overline\epsilon} \frac{\sum_{0\neq \lambda'\in H_1}\pr{\overline{\epsilon+\lambda'}}}{\sum_{\lambda'\in H_1} \pr{\overline{\epsilon+\lambda'}}} 
=1 - \sum_{\overline\epsilon\in\mathcal{E}} \pr{\overline\epsilon} \frac{\pr{\overline{\epsilon}}}{\sum_{\lambda'\in H_1} \pr{\overline{\epsilon+\lambda'}}}.
\end{equation}

We rewrite $\langle \Delta_{\lambda} \rangle$ in the following way
\begin{eqnarray}
\langle \Delta_{\lambda} \rangle &=& \sum_{\epsilon\in C_1} \pr{\epsilon} \Delta_\lambda(\epsilon)
= \sum_{\overline\epsilon\in\mathcal{E}} \pr{\overline\epsilon} \log \frac{\pr{\bar\epsilon}}{\pr{\overline{\epsilon+\lambda}}}\\
&=& \sum_{\overline\epsilon\in\mathcal{E}} \pr{\overline\epsilon} \log \frac{\pr{\bar\epsilon}}{\sum_{\lambda'\in H_1}\pr{\overline{\epsilon+\lambda'}}}
- \sum_{\overline\epsilon\in\mathcal{E}} \pr{\overline\epsilon} \log \frac{\pr{\overline{\epsilon+\lambda}}}{\sum_{\lambda'\in H_1}\pr{\overline{\epsilon+\lambda'}}}.
\end{eqnarray}
Using the inequality $\log x \geq 1 - \frac{1}{x}$ to lower-bound the first term and Jensen inequality for the second term we obtain
\begin{equation}
\langle \Delta_\lambda \rangle \geq 
(1 - | H_1|) - \log\sum_{\overline\epsilon\in\mathcal{E}} \pr{\overline\epsilon} \frac{\pr{\overline{\epsilon+\lambda}}}{\sum_{\lambda'\in H_1} \pr{\overline{\epsilon+\lambda'}}} \rightarrow\infty.
\end{equation}
\end{proof}

\section{Finding phase transitions}

In order to map the disorder-temperature phase diagrams of the $4$- and $6$-body RCIM in Fig.~\ref{fig_phasediag} we need to reliably identify phase transitions. Here we describe in detail how we achieve that by analyzing specific heat, the spin-spin correlation function and the Wilson loop operator. We aslo present additional 
results for the $4$- and $6$-body RCIM in Fig.~\ref{fig_extranumerics}.

\begin{figure*}[h!]
\includegraphics[width=\textwidth]{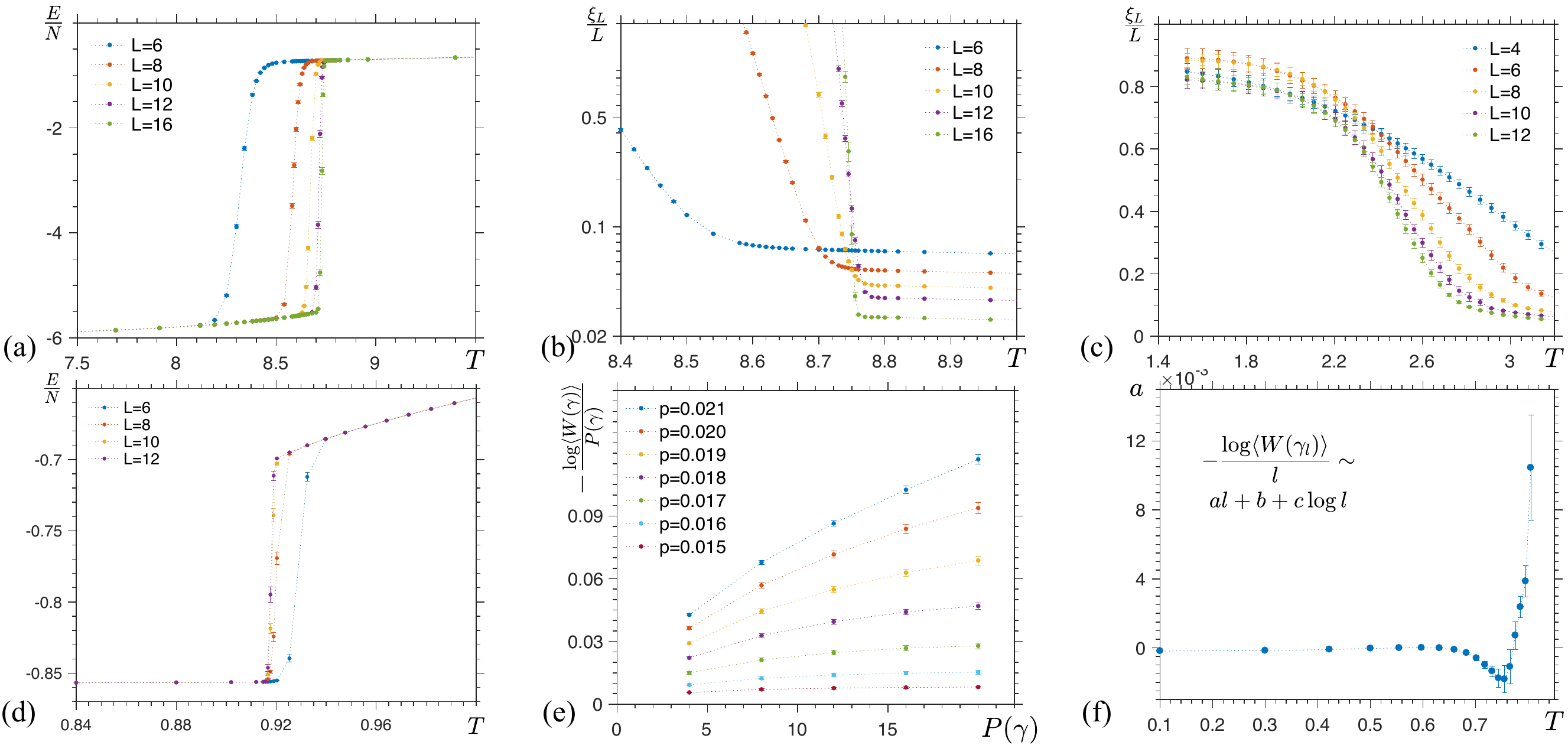}
\caption{
Additional details about the $4$-body (a)-(c) and $6$-body (d)-(f) RCIM.
The discontinuity in energy per spin $E/N$ in (a) and (d) suggests first-order phase transitions for both models for $p=0$.
(b) For $p=0$ the normalized correlation length $\xi_L/L$ does not seem to be described well by the scaling ansatz in Eq.~(\ref{eq_correlation}) possibly due to a transition being first-order.
(c) For the disorder value $p=0.276$ close to the critical point on the Nishimori line $p_N = p^{(2)}_{\mathrm{3DCC}}$ detecting a phase transition and estimating its critical temperature becomes difficult.
(e) We check if the Wilson loop operator $W(\gamma)$ satisfies the perimeter law by plotting $-\log \langle W(\gamma)\rangle /P(\gamma)$ as a function of perimeter $P(\gamma)$ of the square loop $\gamma$ for different disorder values $p$ and fixed temperature $T=0.42$.
We see a change of scaling as the system undergoes a phase transition at $p=0.019(1)$.
(f) We find a fit $-\log \langle W(\gamma_l)\rangle/l \sim al+b+c\log l$ to the data for $p=0.018$ in Fig.~\ref{fig_numerics}(c) and plot the fit coefficient $a$ as a function of temperature $T$.
We identify the critical temperature $T_c = 0.75(3)$ of a transition as a location where $a=0$.
}
\label{fig_extranumerics} 
\end{figure*}

\subsection{Specific heat}

For a second-order phase transition, the specific heat $c(T)$ as a function of temperature $T$ is expected to have a discontinuity near a phase transition at temperature $T_c$ in the limit of infinite system size $L \rightarrow \infty$.
However, for a system of finite linear size $L$, the peak of the specific heat $c_L(T)$ appears at temperature $T_c(L) = \arg\max_T c_L(T)$ shifted from that in the infinite system by an amount
\begin{equation}
\left|\frac{T_L-T_c}{T_c}\right| \propto L^{-1/\nu},
\end{equation}
where $\nu$ is the correlation length critical exponent \cite{Goldenfeld1992}.
A similar scaling behavior has been established for first-order phase transitions \cite{Challa1986,Binder1987,Lee1991,Borgs1992}.
Thus, we find the critical temperature $T_c$ by fitting a function
\begin{equation}
T_c(L) \sim a L^{-b} + T_c
\end{equation}
to the position of the specific heat peaks for different system sizes and evaluating $T_c(L = \infty)$.

\subsection{Correlation function}

One might not be able to identify a phase transition of higher order by looking at the specific heat.
Rather, one needs to analyze the behavior of e.g. the order parameter correlation length $\xi$.
In particular, for the system of finite size $L$ and with fixed disorder strength $p$ we define the two-point finite-size correlation length $\xi_L$ as a function of temperature $T$
\begin{equation}
\xi_L(T)= \frac{1}{2\sin (k_0/2)}\sqrt{\frac{\langle\chi(\vec{0})\rangle}{\langle\chi(\vec{k}_0)\rangle} - 1},
\end{equation}
where $\langle\chi(\vec{k})\rangle = \sum_{\epsilon\subset\face{3}{\mathcal{L}}} \pr{\epsilon}\chi(\vec{k})$, $\vec{k}$ is the wavevector and $\vec{k}_0 = (2\pi/L, 0, 0)$.
In above, we use the thermal expectation value of the wavevector-dependent sublattice magnetic susceptibility
\begin{equation}
\chi(\vec{k}) = \sum_{\{s_v\}} \frac{1}{N} \left(\sum_{u\in U} s_u e^{i \vec{k}\cdot\vec{r_u}} \right)^2 
\frac{e^{-\beta H^X_{\epsilon}(\{s_v\})}}{Z_\epsilon(\beta)}.
\end{equation}
where $\vec{r_u}$ denotes the position of the vertex spin $s_u$ in a sublattice $U\subset \face{0}{\mathcal{L}}$ of single-color vertices.
Near a phase transition at temperature $T_c$, the normalized correlation length is expected to scale as
\begin{equation}
\frac{\xi_L (T)}{L} \sim f(L^{1/\nu} (T-T_c)),
\end{equation}
where $f$ is a dimensionless scaling function and $\nu$ is the correlation length critical exponent.
We can estimate $T_c$ by plotting $\xi_L(T) /L$ as a function of temperature $T$ for different system sizes $L$ and finding their crossing point.
If there is no crossing, then we conclude that there is no phase transition.

\subsection{Wilson loop operator}

When the system under consideration has a local (gauge) symmetry, one cannot use a local order parameter to detect a phase transition.
Rather, one needs to consider gauge-invariant quantities, such as the Wilson loop operator $W(\gamma)$ in Eq.~(\ref{eq_wilson}).
Suppose $\gamma$ is a square loop. 
We denote by $P(\gamma)$ and $A(\gamma)$ the perimeter of $\gamma$ and the minimal area enclosed by $\gamma$, respectively. 
The scaling of the averaged Wilson loop operator $\langle W(\gamma) \rangle$ in the limit of large loops changes between the ordered (Higgs) and disordered (confinement) phases. 
Namely, 
\begin{itemize}
\item in the disordered phase: $\langle W(\gamma)\rangle \sim \exp(-\const\cdot A(\gamma))$,
\item in the ordered phase: $\langle W(\gamma)\rangle \sim \exp(-\const\cdot P(\gamma))$.
\end{itemize}

We consider a system of finite size $L$ and denote by $\gamma_l$ a square loop of linear size $l\leq L/2$.
Since $A(\gamma_l) \propto l^2$ and $P(\gamma_l) \propto l$, then $\log \langle W(\gamma_l) \rangle$ should scale either quadratically or linearly in $l$, depending on the phase of the system.
Due to finite-size effects, there are some corrections to the area and perimeter scaling.
In particular, we numerically find that 
\begin{equation}
- \frac{\log \langle W(\gamma_l)\rangle}{l} \sim a l +b +c\log l,
\label{eq_fitlog}
\end{equation}
where $a,b,c$ are some constants.
We identify the disordered phase as the region where the fitting parameter $a$ is positive, $a>0$.

\section{Classical Ising gauge theory}

As an example of using specific heat and the scaling of the Wilson loop operator to identify a phase transition we study a known model, the three-dimensional random plaquette Ising model (RPIM); see Fig.~\ref{fig_numerics_IGT}.
The RPIM is a generalization of the $\mathbb{Z}_2$ Ising gauge theory, which is relevant for studying the optimal error correction threshold for 1D string-like operators in the 3D toric code \cite{Dennis2002}.
The RPIM is a statistical-mechanical model with classical spins $s_e = \pm 1$ placed on edges $e\in\face{1}{\mathcal{C}}$ of the cubic lattice $\mathcal{C}$ and disorder $\epsilon\subset\face{2}{\mathcal{C}}$.
The Hamiltonian describing the RPIM
\begin{equation}
H^{\mathrm{RPIM}}_\epsilon(\{s_e\}) = -\sum_{f\in\face{2}{\mathcal{C}}} (-1)^{[\epsilon]_f}\ \raisebox{-9pt}{\includegraphics[height = 24pt]{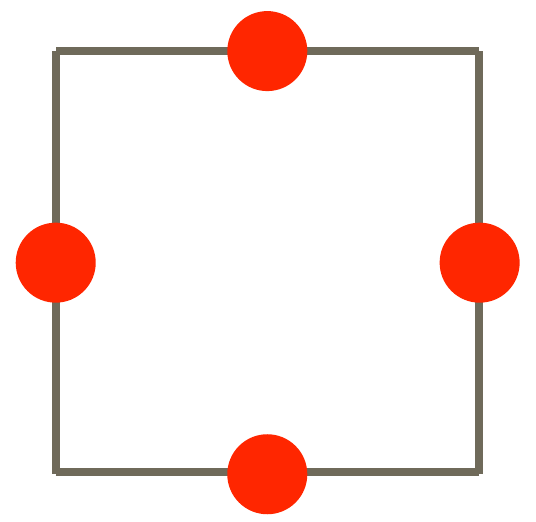}}
\end{equation}
contains $4$-body terms, which are products of four edge spins around every square face $f\in\face{2}{\mathcal{C}}$ of the lattice $\mathcal{C}$.
We set $[\epsilon]_f = 1$ if $f\in\epsilon$, otherwise $[\epsilon]_f = 0$. 
We observe that $H^{\mathrm{RPIM}}_\epsilon(\{s_e\})$ has a local $\mathbb{Z}_2$ symmetry, generated by flips of spins on all edges incident on any vertex $v\in \face{0}{\mathcal{C}}$.
The Wilson loop operator $W(\gamma_l)$ is a gauge-invariant quantity, where $\gamma_l$ is a square loop of linear size $l$.
The disorder-temperature phase diagram of the 3D RPIM is shown in Fig.~\ref{fig_phasediag_IGT}.

\begin{figure*}[h!]
\includegraphics[width=\textwidth]{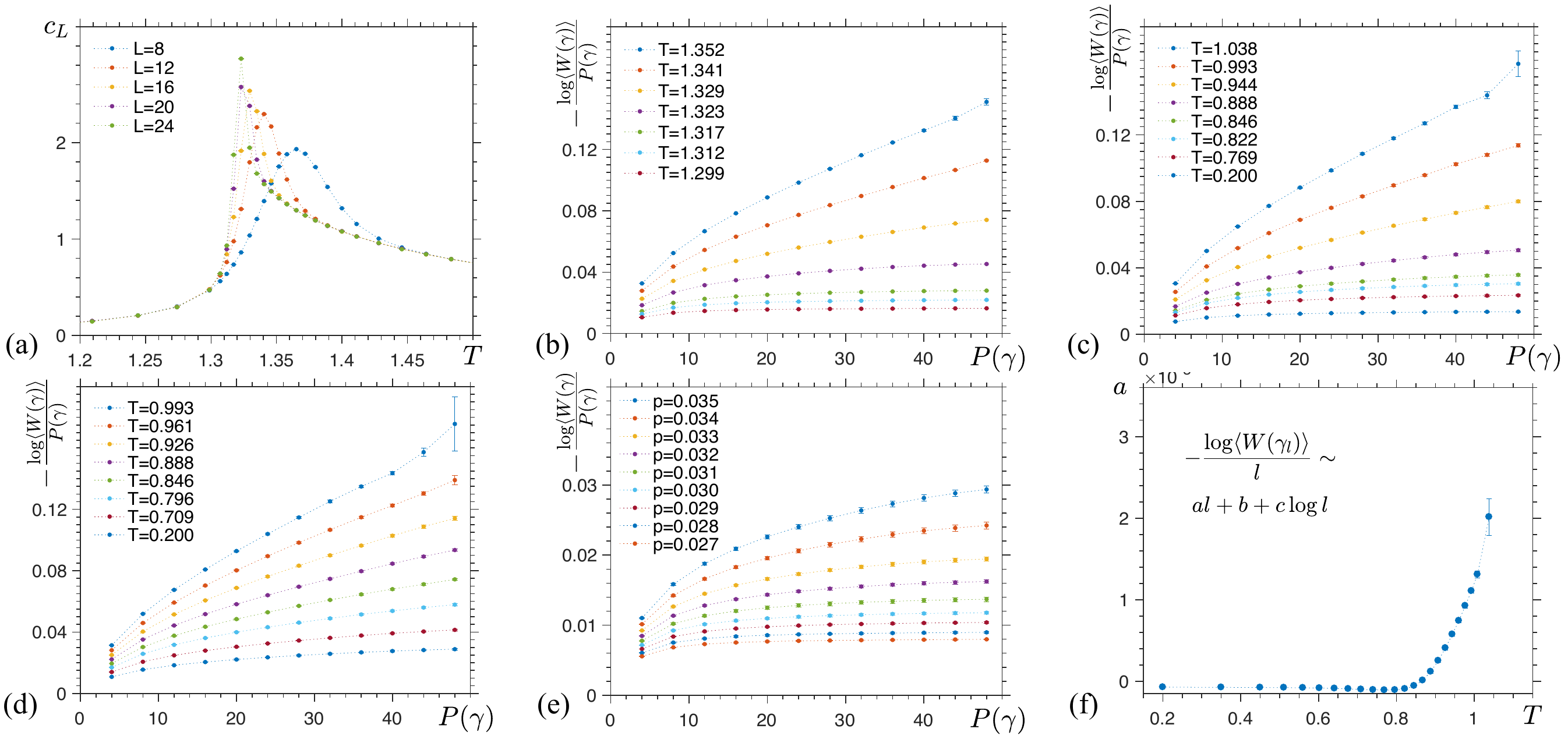}
\caption{
Results for the 3D RPIM. (a) For $p=0$ we can estimate the critical temperature $T_c = 1.316(4)$ of a phase transition by finding the peak positions of specific heat $c_L$ for different system sizes $L$ and exploiting finite-size scaling. 
In (b)-(d) we check for $p=0$, $p=0.031$ and $p=0.035$ whether the Wilson loop operator $W(\gamma)$ satisfies the perimeter law by plotting $-\log \langle W(\gamma)\rangle /P(\gamma)$ as a function of perimeter $P(\gamma)$ of the square loop $\gamma$ for different temperatures $T$.
(e) For fixed temperature $T = 0.45$ we analyze scaling of $-\log \langle W(\gamma)\rangle /P(\gamma)$ for different disorder values $p$.
(f) We find a fit $-\log \langle W(\gamma_l)\rangle/l \sim al+b+c\log l$ to the data in (c) and plot the fit coefficient $a$ as a function of temperature $T$.
We identify the critical temperature $T_c = 0.84(3)$ of a phase transition in (c) as a location where $a=0$.
In (b),(c) and (e) we see a change of scaling as the system undergoes a phase transition at $T_c = 1.317(6)$, $T_c = 0.84(3)$ and $p_c = 0.032(1)$, respectively. In (d) there is no indication of a transition.
}
\label{fig_numerics_IGT} 
\end{figure*}

\begin{figure}[ht!]
\includegraphics[width=0.4\columnwidth]{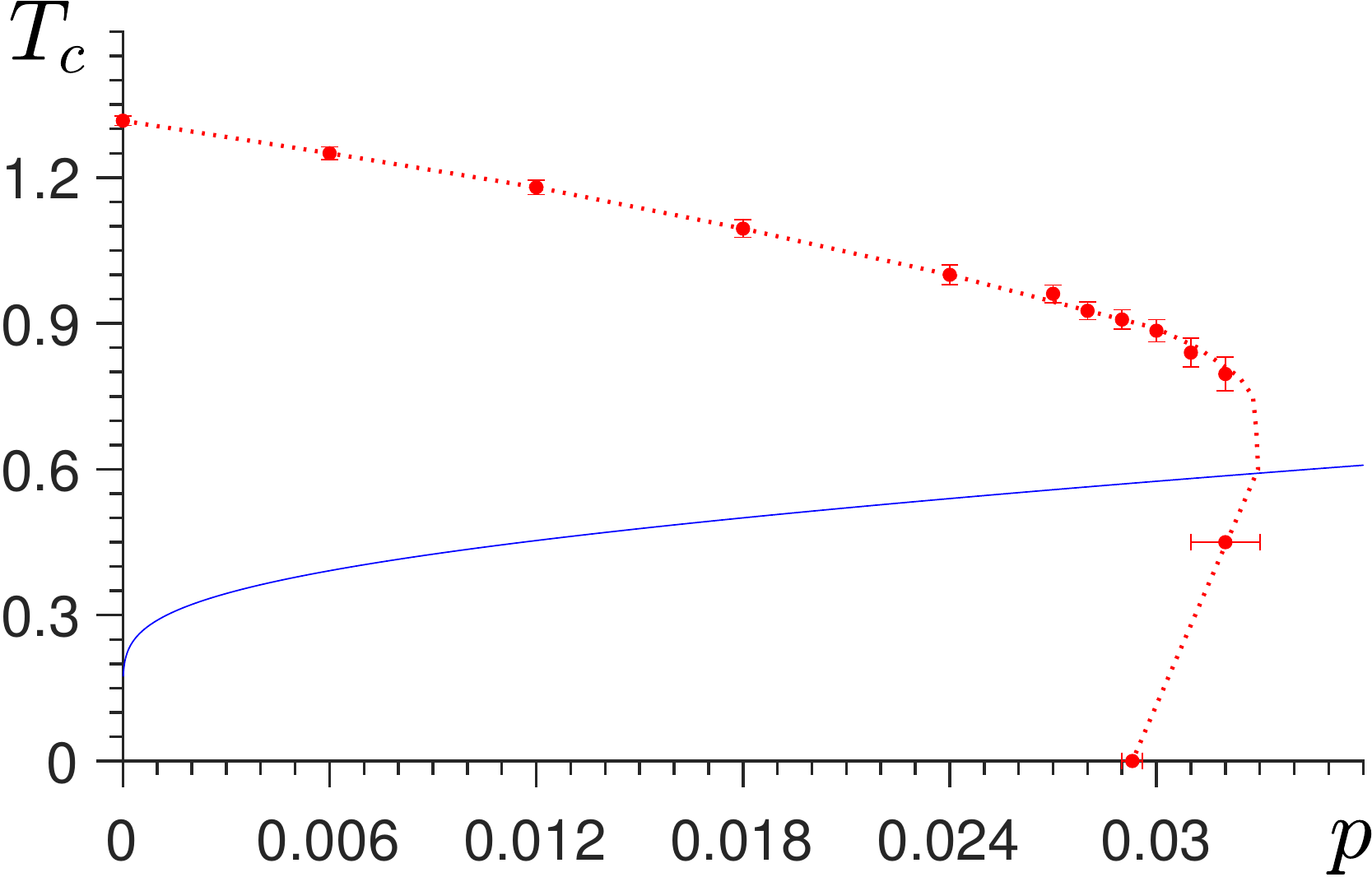}
\caption{
The disorder-temperature $(p,T)$-phase diagram of the 3D random plaquette Ising model on the cubic lattice. 
The intersection of the Nishimori line (blue) with the anticipated phase boundary (red dotted line) gives the 3D toric code threshold $p^{(1)}_{\text{3DTC}} \simeq 3.3\%$ for optimal error correction associated with 1D string-like logical operators (and point-like excitations).
Note that the location of a phase transition for $T=0$ was found in \cite{Wang2002}.}
\label{fig_phasediag_IGT} 
\end{figure}

\newpage
\section{Numerical simulation details}

The numerical complexity of simulating the statistical-mechanical models, such as the $4$- and $6$-body RCIM and the RPIM, increases with the disorder strength $p$, which is reminiscent of a spin glass behavior.
To speed up simulations we use the parallel tempering technique.
The parallel tempering technique requires simultaneous simulation of multiple copies $k=1,\ldots,n$ of the system with the same disorder $\epsilon$ but different spin configurations $\{s_i\}_k$ and temperatures $T_1<\ldots <T_n$.
After performing single-spin Metropolis updates for all spins in every copy of the system, swaps of spin configurations $\{s_i\}_k \leftrightarrow \{s_i\}_{k+1}$ of copies at neighboring temperatures $T_k$ and $T_{k+1}$ are allowed with probability
\begin{equation}
\pr{k \leftrightarrow k+1} = \exp\left( (E_k-E_{k+1})\left(\frac{1}{T_k} - \frac{1}{T_{k+1}}\right) \right),
\end{equation}
where $E_k$ and $E_{k+1}$ denote energies of spin configurations $\{s_i\}_k$ and $\{s_i\}_{k+1}$.
We choose temperatures $T_1<\ldots <T_n$ is such a way that the exchange rate $\{s_i\}_k \leftrightarrow \{s_i\}_{k+1}$ is approximately flat; for more in-depth discussions see e.g.~\cite{Katzgraber2006}.
Equilibration of the system is tested by a logarithmic binning of data.
Numerical simulation details for the $4$-body RCIM, the $6$-body RCIM and the RPIM are presented in Table~\ref{tab_params}.

To estimate statistical error bars of quantities analyzed in the simulation we use the bootstrap technique.
The main idea behind the bootstrap technique is to repeat sampling from the existing data set $D = \{d_1,\ldots, d_N\}$ and evaluating a quantity of interest $q = q(D)$.
In particular, for $i=1,\ldots,n$ we perform the following steps
\begin{enumerate}
\item from the data set $D$ randomly choose $N$ data points $d_{i(j)}$, where $i(j)\in \{1,\ldots, N\}$,
\item evaluate the quantity $q_i = q(D_i)$ from the data set $D_i = \{ d_{i(1)},\ldots,d_{i(N)} \}$.
\end{enumerate}
Note that in step 1 we allow to choose the same data point multiple times.
The relevant quantity $q$ is estimated to be
\begin{equation}
q =  \bar q \pm \sqrt{\sum_{i=1}^n\frac{(\bar q - q_i)^2}{n-1}},
\end{equation}
where $\bar q = \frac{1}{n}\sum_{i=1}^{n}q_i$.

\begin{table}[hp!]
\centering
\begin{tabular*}
{0.5\columnwidth}{@{\extracolsep{\fill} } c c c c c c c}
\hline\hline			
$p$ & $L_{\mathrm{max}}$ & $N_\epsilon$ & $\tau$ & $N_T$& $T_{\mathrm{min}}$ & $T_{\mathrm{max}}$\\
\hline
0.000&	16	&	500&		20&	55&	2.40&	12.80\\
0.050&	16	&	500&		20&	42&	2.30&	11.40\\
0.100&	16	&	500&		20&	41&	2.20&	10.15\\
0.150&	16	&	500&		20&	42&	2.10&	8.42\\
0.200&	16	&	500&		20&	41&  2.00&	6.80\\
0.250&	16	&	500&		20&	42&	1.90&	4.97\\
0.265&	16	&	500&		20&	34&	1.80&	3.53\\
0.270&	16	&	500&		20&	34&	1.60&	3.32\\
0.272&	12	&	500&		20&	34&	1.60&	3.30\\
0.274&	12	&	500&		20&	34&	1.53&	3.21\\
0.276&	12	&	500&		20&	34&	1.53&	3.21\\
0.280&	12	&	500&		20&	34&	1.33&	3.18\\
\hline\hline
0.000&	12		&	250&		20&	47&	0.20&	1.28\\
0.003&	12		&	250&		20&	44&	0.20&	1.25\\
0.006&	12		&	250&		20&	42&	0.20&	1.22\\
0.009&	12		&	250&		20&	39&	0.20&	1.17\\
0.012&	12		&	250&		20&	38&	0.20&	1.14\\
0.015&	10		&	250&		21&	48&	0.10&	1.35\\
0.016&	10		&	250&		21&	48&	0.10&	1.35\\
0.017&	10		&	250&		21&	48&	0.10&	1.35\\
0.018&	10		&	250&		21&	48&	0.10&	1.35\\
0.019&	10		&	250&		21&	48&	0.10&	1.35\\
0.020&	10		&	250&		21&	48&	0.10&	1.35\\
0.021&	10		&	250&		21&	48&	0.10&	1.35\\
\hline\hline
0.000&	24	&	500&		19&	51&	0.40&	2.08\\
0.006&	24	&	500&		19&	43&	0.40&	1.95\\
0.012&	24	&	500&		19&	41&	0.40&	1.77\\
0.018&	24	&	500&		19&	43&	0.35&	1.64\\
0.024&	24	&	500&		19&	42&	0.30&	1.49\\
0.027&	24	&	250&		19&	43&	0.20&	1.28\\
0.028&	24	&	250&		19&	43&	0.20&	1.28\\
0.029&	24	&	250&		19&	43&	0.20&	1.28\\
0.030&	24	&	250&		19&	43&	0.20&	1.28\\
0.031&	24	&	250&		19&	43&	0.20&	1.28\\
0.032&	24	&	250&		19&	43&	0.20&	1.28\\
0.023&	24	&	250&		19&	43&	0.20&	1.28\\
0.034&	24	&	250&		19&	43&	0.20&	1.28\\
0.035&	24	&	250&		19&	43&	0.20&	1.28\\
\hline\hline
\end{tabular*}
\caption{
Numerical simulation parameters for: the $4$-body RCIM (top), the $6$-body RCIM (middle), and the IGT (bottom).
$L_{\mathrm{max}}$ and $N_\epsilon$ denote the linear size of the biggest simulated system and the number of randomly chosen disorder samples.
$N_T$ denotes the number of temperatures in the range $[T_{\mathrm{min}},T_{\mathrm{max}}]$ chosen in a way that the exchange rate of spin configurations is approximately constant.
$2^\tau$ is the number of equilibration steps, where one equilibration step consists of an update of every spin in all $N_T$ copies of the system followed by swaps $\{s_i\}_k \leftrightarrow \{s_i\}_{k+1}$ of spin configurations. }
\label{tab_params}
\end{table}

\clearpage
\bibliography{biblio}

\begin{thebibliography}{47}
\bibitem [{sub()}]{subsystem}%
  \BibitemOpen
  \href@noop {} {}\bibinfo {note} {The 3D subsystem color code has the gauge
  group $\mathcal{G} = \langle\prod_{\face{3}{\mathcal{L}}\ni t\supset e} X(t),
  \prod_{\face{3}{\mathcal{L}}\ni t\supset e} Z(t) | e\in\face{1}{\mathcal{L}}
  \rangle$ and the stabilizer group $\mathcal{S} =
  \langle\prod_{\face{3}{\mathcal{L}}\ni t\supset v} X(t),
  \prod_{\face{3}{\mathcal{L}}\ni t\supset v} Z(t) | v\in\face{0}{\mathcal{L}}
  \rangle$.}\BibitemShut {Stop}%
\bibitem [{com()}]{comparison}%
  \BibitemOpen
  \href@noop {} {}\bibinfo {note} {Note that this threshold was estimated for
  the code on a different 3D lattice and thus the comparison should only be
  considered as a guide to intuition.}\BibitemShut {Stop}%
\end{thebibliography}

\end{document}